%% file: QSFW_rand.tex
\documentclass[pra,twocolumn,showpacs]{revtex4-2}

\usepackage{amsmath}
\usepackage{latexsym}
\usepackage{amssymb}
\usepackage{graphicx}
\usepackage[colorlinks=true, citecolor=blue, urlcolor=blue]{hyperref}
\usepackage{float}
\usepackage{amsfonts}
\usepackage{textcomp}
\usepackage{mathpazo}
\usepackage{comment}
\usepackage{xr}

\usepackage{bbm}

\usepackage{xcolor}
\definecolor{myurlcolor}{rgb}{0,0,0.4}
\definecolor{mycitecolor}{rgb}{0,0.5,0}
\definecolor{myrefcolor}{rgb}{0.5,0,0}
\usepackage{hyperref}
\hypersetup{colorlinks,
linkcolor=myrefcolor,
citecolor=mycitecolor,
urlcolor=myurlcolor}

\usepackage[draft]{fixme}
\usepackage{amsmath,bbm}
\usepackage{graphicx}
\usepackage{amsfonts}
\usepackage{amssymb}
\usepackage{amsmath, amssymb, amsthm,verbatim,graphicx,bbm}
\usepackage{mathrsfs}
\usepackage{color,xcolor,longtable}

\usepackage{xr}
\makeatletter
\newcommand*{\addFileDependency}[1]{
  \typeout{(#1)}
  \@addtofilelist{#1}
  \IfFileExists{#1}{}{\typeout{No file #1.}}
}
\makeatother

\newcommand*{\myexternaldocument}[1]{
    \externaldocument{#1}
    \addFileDependency{#1.tex}
    \addFileDependency{#1.aux}
}

\myexternaldocument{manuscript_PRL}

\newcommand{\beq}[0]{\begin{equation}}
\newcommand{\eeq}[0]{\end{equation}}

\newcommand{\one}{\leavevmode\hbox{\small1\normalsize\kern-.33em1}}

\def\be{\begin{equation}}
\def\ee{\end{equation}}
\def\ben{\begin{eqnarray}}
\def\een{\end{eqnarray}}
\def\eea{\end{array}}
\def\bea{\begin{array}}

\newcommand{\Tr}[1]{\mathrm{Tr}#1}
\newcommand{\bei}{\begin{itemize}}
\newcommand{\eei}{\end{itemize}}
\newcommand{\ket}[1]{|#1\rangle}
\newcommand{\bra}[1]{\langle#1|}

\newcommand{\proj}[1]{\ket{#1}\!\!\bra{#1}}

\newcommand{\I}{\mathbbm{1}}
\newcommand{\w}{\omega}

\renewcommand{\emph}[1]{\textbf{#1}}

\makeatletter
\newtheorem*{rep@theorem}{\rep@title}
\newcommand{\newreptheorem}[2]{%
\newenvironment{rep#1}[1]{%
 \def\rep@title{#2 \ref{##1}}%
 \begin{rep@theorem}}%
 {\end{rep@theorem}}}
\makeatother

\theoremstyle{plain}
\newtheorem{thm}{Theorem}
\newtheorem*{thm*}{Theorem}
\newreptheorem{thm}{Theorem}

\theoremstyle{definition}

\theoremstyle{remark}

\usepackage[T1]{fontenc}

\begin{document}

\title{Certification of randomness without seed randomness}
\author{Shubhayan Sarkar}
\email{shubhayan.sarkar@ulb.be}
\affiliation{Laboratoire d’Information Quantique, Université libre de Bruxelles (ULB), Av. F. D. Roosevelt 50, 1050 Bruxelles, Belgium}

\begin{abstract}	
The security of any cryptographic scheme relies on access to random number generators. Device-independently certified random number generators provide maximum security as one can discard the presence of an intruder by considering only the statistics generated by these devices. Any of the known device-independent schemes to certify randomness require an initial feed of randomness into the devices, which can be called seed randomness. In this work, we propose a one-sided device-independent scheme to certify two bits of randomness without the initial seed randomness. For our purpose, we utilise the framework of quantum networks with no inputs and two independent sources shared among two parties with one of them being trusted. Along with it, we also certify the maximally entangled state and the Bell basis measurement with the untrusted party which is then used to certify the randomness generated from the untrusted device.

\end{abstract}


\maketitle

\section{Introduction}

Device-independent (DI) certification of quantum states and measurement is a type of verification of an unknown quantum device without relying on any assumptions about the device's internal workings. 
Such a certification is achieved by comparing the device's input-output statistics to the expected statistics of an ideal device. This allows one to make inferences about the device's underlying properties to be equivalent to the ideal one's properties up to some degree of freedom. 
The strongest DI certification is termed self-testing where one can completely characterize the states and measurements up to some degrees of freedom.
An important application of DI certification is towards ensuring that one can securely generate genuine randomness from the outcomes of the measurement devices even when an intruder might have access to them. This is extremely important for any cryptographic scheme as the security of these schemes relies on access to random number generators.

Any of the known schemes for DI certification of randomness requires access to seed randomness, that is, the measurement devices whose outcomes will be used to generate random numbers, have inputs that have to be chosen randomly in order for the protocol to be secure [for instance see Refs. \cite{di4, random0, random1, rand1, rand2, rand3, Fehr, Pironio2, APP13, Armin1, chainedBell, sarkar, sarkar5}]. The reason is that most of these schemes rely on violating a Bell inequality \cite{Bell66} which requires that the measurement devices have at least two inputs that are randomly sampled. Even when considering relaxations by assuming some properties about the devices involved, known as semi-DI schemes, 
one still requires access to seed randomness in order to certify genuine random numbers [for instance see Refs. \cite{Armin3, Farkas, PMrand1, PMrand2, PMrand3, PMrand4, PMrand5, PMrand6, PMrand7}].     
One such class of semi-DI schemes are known as one-sided DI (1SDI), where one of the parties is assumed to be trusted, that is, the measurements performed by this party are known. There are a few works \cite{steerand2, steerand3, steerand4, sarkar12} that propose 1SDI schemes for the certification of randomness generated from the measurement outcomes of the untrusted device. In particular, \cite{sarkar12} proposes a 1SDI scheme that can be used to certify the maximum amount of randomness, that is $2\log_2d$ bits, extractable from a $d-$dimensional system with the trusted party having $d+1$ inputs along with three inputs for the untrusted party. 

Here, we consider the problem of certifying randomness without the initial feed of randomness to the measurement devices. For this purpose, we consider the framework of quantum networks introduced in \cite{pironio1, pironio21, Fritz}. Particularly in Refs. \cite{pironio21,Fritz}, it was observed that by considering independent sources that are shared among non-communicating parties one can detect quantum nonlocality using only a single fixed measurement for each party. Further on, DI certification of quantum states and measurements in quantum networks was recently explored in Refs. \cite{Marco, NLWEsupic, JW2, Allst1, supic4, sekatski, sarkar2023}. However, all of these certification schemes require at least two inputs for most of the measurement devices. A partial certification scheme was proposed in \cite{sekatski} that utilizes the genuine network nonlocality without inputs in a triangle network \cite{renou1}. However, using the proposed scheme \cite{sekatski}, one can only conclude that the sources need to prepare entangled states and the measurements performed by the parties are non-classical. Thus, the exact certification of quantum states, measurements and randomness without inputs in a DI way or semi-DI way is still an open problem.

Recently, a form of quantum nonlocality in networks, termed as swap-steering, was introduced in \cite{sarkar15}. The scenario consisted of two spatially separated parties, with one of them being trusted. Each of them receives two subsystems from two independent sources on which they perform a single fixed measurement. 
In order to certify randomness, we first certify the states and measurements in the presence of an adversary who might have access to them. For our purpose, we consider the witness introduced in \cite{sarkar15} and use it for self-testing 
two maximally entangled states and the Bell basis measurement with the untrusted party up to the freedom of local unitaries and the presence of some junk Hilbert space. Then, we show that the measurement outcomes of the untrusted party are genuinely random as the adversary can not guess the outcomes of this measurement given the assumption that the sources can be correlated only in a classical way. Thus, we are able to certify the amount of two bits of randomness without any initial feed of randomness to the devices in a 1SDI way.  

Before proceeding with the result, let us first introduce the relevant concepts required for the manuscript. 

\section{Preliminaries}

Let us briefly describe the swap-steering scenario introduced in \cite{sarkar15} which involves two parties named Alice and Bob, who are located in separate labs. They each receive two subsystems from two distinct sources, denoted as $S_1$ and $S_2$, which are statistically independent of one another. They then conduct a solitary four-outcome measurement on their respective subsystems, with the outcomes labeled as $a,b=0,1,2,3$ for Alice and Bob, respectively [see Fig. \ref{fig1}]. It should be noted that Alice is considered trustworthy in this scenario, which means that the measurement on her subsystems is well-known which in this case corresponds to the Bell basis given by $ M_{A}=\{\proj{\phi_{+}},\proj{\phi_{-}},\proj{\psi_{+}},\proj{\psi_{-}}\}_{A_1A_2} $ where

\begin{eqnarray}\label{Amea1}
    \ket{\phi_{\pm}}_{A_1A_2}&=&\frac{1}{\sqrt{2}}\left(\ket{0}_{A_1}\ket{0}_{A_2}\pm\ket{1}_{A_1}\ket{1}_{A_2}\right)\nonumber\\
    \ket{\psi_{\pm}}_{A_1A_2}&=&\frac{1}{\sqrt{2}}\left(\ket{0}_{A_1}\ket{1}_{A_2}\pm\ket{1}_{A_1}\ket{0}_{A_2}\right).
\end{eqnarray}
Here $A_i/B_i\ (i=1,2)$ denote the two different subsystems of Alice/Bob respectively. Alice and Bob proceed to repeat the experiment multiple times to create the joint probability distribution (also known as correlations) represented by $\vec{p}={p(a,b)}$. Here, $p(a,b)$ refers to the likelihood of Alice and Bob obtaining the outcomes $a,b$, respectively.

In quantum theory, it is advantageous to express the correlations in terms of expectation values rather than probability distributions. When dealing with $d$-outcome measurements, a useful technique is to utilize the two-dimensional Fourier transform of the conditional probabilities $p(a,b)$ as
\begin{equation}\label{ExpValues}
    \langle A^{(k)}_0 B^{(l)}_0 \rangle = \sum^{d-1}_{a,b=0} \w^{ak+bl} p(a,b),
\end{equation}
where $\w$ is the $d$-th root of unity $\omega=\exp(2\pi\mathbbm{i}/d)$ and $k,l=0,\ldots,d-1$ and $ A^{(k)}_0,B^{(l)}_0$ are known as observables. 
Using the inverse Fourier transform of \eqref{ExpValues}, we obtain that
\begin{eqnarray}\label{ExpValues1}
    p(a,b)=\frac{1}{d^2}\sum_{k,l=0}^{d-1}\omega^{-(ak+bl)} \langle A^{(k)}_0 B^{(l)}_0 \rangle.
\end{eqnarray}

The expectation value appearing on the left-hand side of Eq. (\ref{ExpValues}) can be simply represented as $\langle A^{(k)}_0 B^{(l)}_0 \rangle = \Tr(A^{(k)}_0 \otimes B^{(l)}_0\rho_{AB})$ for some state $\rho_{AB}$ with $\{A^{(k)}_0\}$ and $\{B^{(l)}_0\}$ are operators defined as
\begin{equation}\label{obsgen}
    A^{(k)}_0= \sum^{d-1}_{a=0} \w^{ak} P^{(a)}, \qquad B^{(l)}_0 = \sum^{d-1}_{b=0} \w^{bl} Q^{(b)}.
\end{equation}
where $P^{(a)},Q^{(b)}$ represent the measurement elements of Alice, Bob respectively.
As proven in \cite{Jed1}, the observables $ A^{(k)}_0$ have the following properties (same for $B^{(l)}_0$): $A^{(d-k)}_0=(A^{(k)}_0)^{\dagger}$ and $A^{(k)}_0(A^{(k)}_0)^{\dagger}\leq\I$. For the special case of projective measurements, the observables $A^{(k)}_0$ are unitary and $ A^{(k)}_0= (A^{(1)}_0)^k=A^k_0$.

As described above, Alice performs the Bell-basis measurement whose corresponding measurement elements for the rest of the manuscript will be denoted as $\ket{\phi_1}=\ket{\phi^+},\ket{\phi_2}=\ket{\phi^-},\ket{\phi_3}=\ket{\psi^+}, \ket{\phi_4}=\ket{\psi^-}$ and the corresponding observable using \eqref{obsgen} is given as
\begin{eqnarray}\label{A0}
    A_0=\sum_{k=0}\mathbbm{i}^{k}\proj{\phi_k}.
\end{eqnarray}
Let us now define the task of self-testing relevant to this manuscript. 

\begin{figure}[t]
\includegraphics[width=\linewidth]{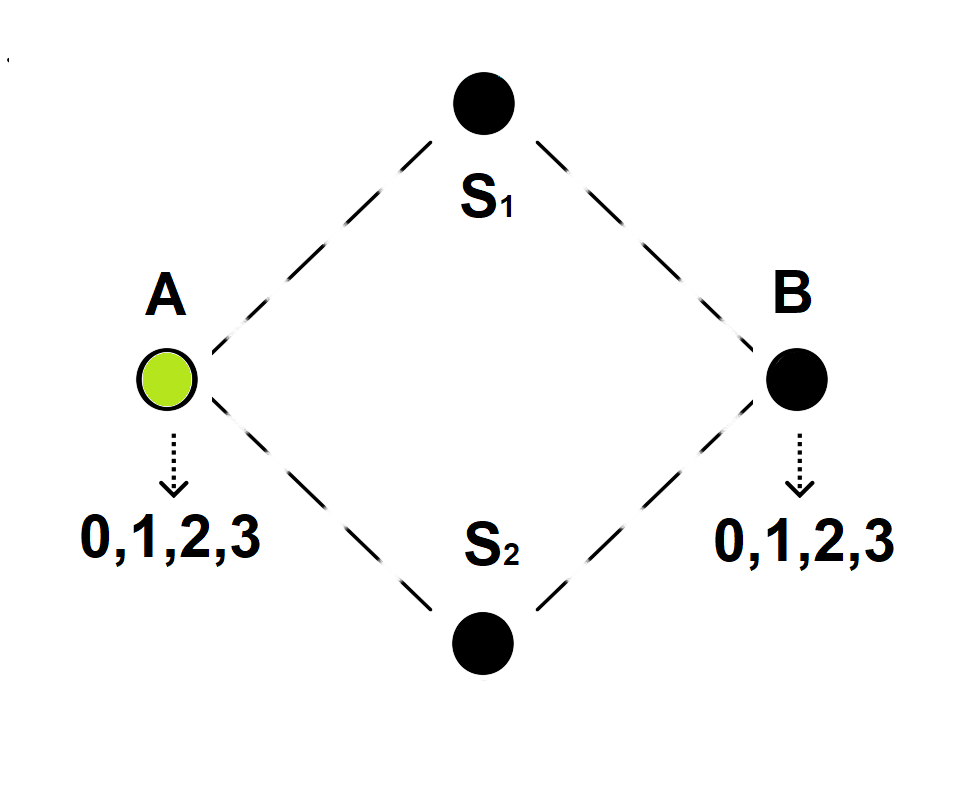}
    \caption{Swap-steering scenario. Alice and Bob are physically distant from each other and receive two subsystems each from the sources $S_1, S_2$. Both of them independently conduct a single four-outcome measurement on their respective received subsystems. It is essential to note that Alice is considered trustworthy and is known to perform the Bell-basis measurement. Communication between them during the experiment is strictly prohibited.}
    \label{fig1}
\end{figure}

{\it{Self-testing in 1SDI scenario without inputs.}} Self-testing in the 1SDI scenario was first defined 
in Ref. \cite{Supic, Alex}. Inspired by \cite{sarkar6, sarkar12, sarkar11}, we present a general definition of self-testing in the 1SDI scenario in quantum networks without inputs with one trusted party. Interestingly, we do not require assuming a pure underlying state or projective measurements.

Let us revisit the previous experiment in which Alice and Bob conduct measurements on the states $\rho_{AB}$ prepared by the sources $S_i\ (i=1,2)$ and observe the correlations ${p(a,b)}$. It is important to note that $A_0$ is fixed, whereas Bob's observables $B_0^{(k)}$ is arbitrary. Now, let's examine a reference experiment that reproduces the same statistics as the actual experiment but involves the states $\tilde{\rho}_{AB}$ and observables represented by $\tilde{B}_0^{(k)}$, which both parties wish to validate. The states $\rho_{AB}$ and the observables $B_0^{(k)}$ are self-tested from $\{p(a,b)\}$ if there exists a unitary $U_B:\mathcal{H}_B\to \mathcal{H}_B$ such that 
\begin{equation}
 (\mathbbm{1}_A\otimes U_B)\rho_{AB}(\mathbbm{1}_A\otimes U_B^{\dagger})=\tilde{\rho}_{AB'}\otimes\rho_{B''},
\end{equation}
\begin{equation}
    U_B\,B_0^{(k)}\,U_B^{\dagger}=\tilde{B}_0^{(k)}\otimes\mathbbm{1}_{B''},
\end{equation}
where $\mathcal{H}_B$ decomposes as $\mathcal{H}_B=\mathcal{H}_{B'}\otimes\mathcal{H}_{B''}$ such that $\mathcal{H}_{B''}$ denotes the junk Hilbert space. The states $\rho_{B''}$ and $\mathbbm{1}_{B''}$ denote the junk state and the identity acting on $\mathcal{H}_{B''}$.

\section{Results}

Let us first revisit the swap-steering inequality introduced in \cite{sarkar15}
\begin{eqnarray}\label{steein}
W=p(0,0)+p(1,1)+p(2,2)+p(3,3)\leq \beta_{LHS}
\end{eqnarray}
Using \eqref{ExpValues1}, the above steering inequality can be simply represented as
\begin{eqnarray}\label{steein1}
  W= \frac{1}{4}\sum_{k=0}^3 \langle A_0^k\otimes B_0^{(4-k)} \rangle\leq \beta_{LHS}
\end{eqnarray}
As shown in \cite{sarkar15}, the quantum bound of the above steering inequality is $1$ which is also the maximum algebraic value of $W$. Consequently, we observe from \eqref{steein1} that the maximum value can be attained iff each term is $1$, that is, for $k=0,1,2,3$
\begin{eqnarray}
    \langle A_0^k\otimes B_0^{(4-k)} \rangle=1.
\end{eqnarray}
Now, using Cauchy-Schwarz inequality we get that
\begin{eqnarray}\label{SOS3}
     A_0^k\otimes B_0^{(4-k)}\rho_{AB}=\rho_{AB}.
\end{eqnarray}
Recalling that $\rho_{AB}$ is separable, we can express it as $\rho_{AB}=\sum_jp_j\ \rho^j_{A_1B_1}\otimes\rho^j_{A_2B_2}$ which using its eigendecomposition can be expressed as $\rho_{AB}=\sum_{s,s'}p_{s,s'}\proj{\psi_{s,A_1B_1}}\otimes\proj{\psi_{s',A_2B_2}}$. Consequently, we get from the above expression Eq. \eqref{SOS2} that
\begin{equation}\label{SOS1}
    \sum_{s,s'}p_{s,s'} A_0^k\otimes B_0^{(4-k)}\ \psi^1_s\otimes\psi_{s'}^2=\sum_{s,s'}p_{s,s'}\ \psi^1_s\otimes\psi_{s'}^2
\end{equation}
where for simplicity, we represent the states $\proj{\psi_{s,A_iB_i}}$ as $\psi^i_s$.
It is now straightforward to observe from the above relation that for all $s,s'$
\begin{eqnarray}\label{SOS2}
     A_0^k\otimes \overline{B}_{0,ss'}^{4-k}\ \ket{\psi^1_s}\ket{\psi^2_{s'}}=\ket{\psi^1_s}\ket{\psi^2_{s'}}
\end{eqnarray}
Here $\overline{B}_{0,ss'}$ is the projection of $B_0$ on the support of $\Tr_A\psi^1_s\otimes\Tr_A\psi^2_{s'}$.
The above relations are sufficient to self-test the state $\rho_{AB}$ and Bob's measurement $B_0$. Before proceeding toward the self-testing result, it is important to recall the assumption that the local states are full-rank as the measurements can only be characterized on the local support of the states.
For a note, we closely follow the techniques introduced in \cite{sarkar6}. 

\begin{thm}\label{Theo1M} 
Assume that the steering inequality \eqref{steein1},  with trusted Alice choosing the observable $A_0$ \eqref{A0}, is maximally violated by a separable state $\rho_{AB}$ acting on $\mathbbm{C}^2\otimes\mathbbm{C}^2\otimes\mathcal{H}_B$ and Bob's observable $B_0$. Then, the following statements hold true:
\\
\\
1. Bob's measurement is projective with his Hilbert space decomposing as $\mathcal{H}_{B}=(\mathbbm{C}^2)_{B_1'}\otimes(\mathbbm{C}^2)_{B_2'}\otimes \mathcal{H}_{B_{12}''}$ for some auxiliary Hilbert space $\mathcal{H}_{B''_{12}}=\mathcal{H}_{B''_{1}}\otimes \mathcal{H}_{B''_{2}}$.\\
\\
2.  \ \  There exist unitary transformations, $U_{i}:\mathcal{H}_B\rightarrow\mathcal{H}_B$,  such that
\begin{eqnarray}\label{lem1.2}
(\mathbbm{1}_{A}\otimes U_B)\rho_{AB}(\mathbbm{1}_{A}\otimes U_B^{\dagger})\qquad\qquad\qquad\qquad\nonumber\\=\proj{\phi^+}_{A_1B_1'}\otimes\proj{\phi^+}_{A_2B_2'}\otimes \rho_{B_1''B_2''},
\end{eqnarray}
where $B_i''$ denotes Bob's auxiliary system, and 
\begin{eqnarray}\label{lem1.1}
\quad U_B\,B_0\,U_B^{\dagger}=A_0\otimes \mathbbm{1}_{B_1''B_2''}
\end{eqnarray}
where $U_B=U_1\otimes U_2$.
\end{thm}
\begin{proof}
Let us first show that Bob's measurement is projective. For this purpose, we consider the relations \eqref{SOS3} for $k=1$ and then multiply it with $A_0^3\otimes B_0$ to obtain
\begin{eqnarray}\label{Bobmea1}
    \I_A\otimes B_0B_0^{(3)}\ \rho_{AB}=A_0^3\otimes B_0\ \rho_{AB}
\end{eqnarray}
where we used the fact that $A_0^4=\I_A$. 
Notice that the right-hand side of the above expression \eqref{Bobmea1} can be simplified using the relation \eqref{SOS3} for $k=3$ to obtain
\begin{eqnarray}
     \I_A\otimes B_0B_0^{(3)}\ \rho_{AB}= \rho_{AB}.
\end{eqnarray}
Thus, taking a partial trace over Alice's subsystem and recalling that $B_0^{(3)}=B_0^{\dagger}$ gives us
\begin{eqnarray}
B_0B_0^{\dagger}\ \rho_B=\rho_B
\end{eqnarray}
where $\rho_B=\Tr_B\  \rho_{AB}$. As the local states are full-rank, they are invertible too and consequently one can arrive at
\begin{eqnarray}
    B_0B_0^{\dagger}=\I_B.
\end{eqnarray}
Similarly, one can also find that $B_0^{\dagger}B_0=\I_B$. Both these relations of Bob's observable suggest that the observable $B_0$ and unitary, and thus Bob's measurement is projective. In a similar manner, considering the relation \eqref{SOS2} one can observe that  $\overline{B}_{0,ss'}$ for all $s,s'$ are unitary.

Let us now consider the relation Eq. \eqref{SOS2} and characterize the states $\ket{\psi^1_s},\ket{\psi^2_{s'}}$ that satisfy the relation \eqref{SOS2}. For simplicity, we drop the indices $s,s'$ for now.
As the local states on Alice's side belong to $\mathbb{C}^2$, using Schmidt decomposition we represent  $\ket{\psi^1},\ket{\psi^2}$ as
\begin{eqnarray}
\ket{\psi^i}=\sum_{j=0,1}\lambda_{j,i}\ket{e_{j,i}}\ket{f_{j,i}}
\end{eqnarray}
where $\lambda_{j,i}\geq0$ and $\{\ket{e_{j,i}}\},\{\ket{f_{j,i}}\}$ form an orthonormal basis for each $i$.
Now applying a unitary $U_{i}$ on these states such that $U_{i}\ket{f_{j,i}}=\ket{e^*_{j,i}}$ gives us
\begin{eqnarray}
\ket{\tilde{\psi}^i}=U_{i}\ket{\psi^i}=\sum_{j=0,1}\lambda_{j,i}\ket{e_{j,i}}\ket{e^*_{j,i}}.
\end{eqnarray}
Now, notice that the state on the right-hand side can be represented as
\begin{eqnarray}\label{state1}
\ket{\tilde{\psi}^i}= P_{i}\otimes\I_{B_i}\ket{\phi^+}
\end{eqnarray}
where
\begin{eqnarray}
    P_i=\sqrt{d}\sum_{j=0,1}\lambda_{j,i}^2\proj{e_{j,i}}.
\end{eqnarray}
Notice that $P_i$ is full-rank as states that are separable between Alice and Bob can not violate the swap-steering inequality \eqref{steein}.
Putting the state \eqref{state1} in the relation \eqref{SOS2} gives us
\begin{eqnarray}\label{st22}
    A_0(P_1\otimes P_2) \otimes \tilde{B}_0^{\dagger}\ket{\phi^+}\ket{\phi^+}=P_1\otimes P_2 \ket{\phi^+}\ket{\phi^+}
\end{eqnarray}
where $\tilde{B}_0=U_{1}^{\dagger}\otimes U_2^{\dagger}\ \overline{B}_0\ U_{1}\otimes U_2$.
Now, using the fact that
\begin{eqnarray}
\ket{\phi^+}_{A_1B_1}\ket{\phi^+}_{A_2B_2}=\ket{\phi^+_4}_{A_1A_2|B_1B_2}
\end{eqnarray}
where $\ket{\phi^+_4}$ is the maximally entangled state of local dimension four. This allows us to conclude from \eqref{st22} that
\begin{eqnarray}
   (P_1^{-1}\otimes P_2^{-1}) A_0(P_1\otimes P_2) \otimes \tilde{B}_0^{\dagger}\ \ket{\phi^+_4}=\ket{\phi^+_4}.
\end{eqnarray}
Now, using the fact that $R\otimes Q\ket{\phi^+}=RQ^T\otimes \I\ket{\phi^+}$, where $T$ denotes the transpose in the computational basis, gives us
\begin{eqnarray}
     (P_1^{-1}\otimes P_2^{-1}) A_0(P_1\otimes P_2) \tilde{B}_0^{*}\otimes\I_B\ket{\phi^+_4}=\ket{\phi^+_4}.
\end{eqnarray}
Taking the partial trace over $B's$ subsystem allows us to conclude that
\begin{eqnarray}
     (P_1^{-1}\otimes P_2^{-1}) A_0(P_1\otimes P_2) \tilde{B}_0^{*}=\I_A
\end{eqnarray}
which eventually leads us to Bob's measurement being
\begin{eqnarray}\label{Bmea1}
    \tilde{B}^T_0= (P_1^{-1}\otimes P_2^{-1}) A_0(P_1\otimes P_2).
\end{eqnarray}
As $\tilde{B}_0$ is unitary and $P_1,P_2$ are Hermitian, we get from the above condition that
\begin{eqnarray}
     (P_1^{-1}\otimes P_2^{-1}) A_0(P_1\otimes P_2)^2A_0^{\dagger}(P_1^{-1}\otimes P_2^{-1})=\I_A.
\end{eqnarray}
Rearranging the terms we obtain that
\begin{eqnarray}
     A_0(P_1\otimes P_2)^2=(P_1\otimes P_2)^2A_0
\end{eqnarray}
which is equivalent to 
\begin{eqnarray}
    [A_0,(P_1\otimes P_2)^2]=0.
\end{eqnarray}
Now, notice that if two matrices commute then they share the same basis. However, the matrix $A_0$ has an entangled basis and the matrix $P_1\otimes P_2$ have a product basis. Thus, the only instance for these two matrices to commute is when $P_1\otimes P_2=\I$ which imposes that $P_1=P_2=\I$. Going back to Eq. \eqref{state1} allows us to conclude that the states $\ket{\psi^1},\ket{\psi^2}$ are the maximally entangled state, that is,
\begin{eqnarray}
   \I_A\otimes U_{i}\ket{\psi^i}= \ket{\phi^+}\quad i=1,2
\end{eqnarray}
and Bob's measurement using \eqref{Bmea1} is
\begin{eqnarray}
  U_{1}^{\dagger}\otimes U_2^{\dagger}\ \overline{B}_0\ U_{1}\otimes U_2=A_0^T=A_0.
\end{eqnarray}

Let us now bring back the indices $s,s'$ and rewrite the states and measurements as
\begin{eqnarray}
\ket{\psi^i_{s}}=\frac{1}{\sqrt{2}}\sum_{j=0,1}\ket{j}\ket{f_{j,i,s}}
\end{eqnarray}
where $U_{s,i}^{\dagger}\ket{j}=\ket{f_{j,i,s}}$ and 
\begin{eqnarray}
     \overline{B}_{0,ss'}=U_{s,1}\otimes U_{s',2}\ A_0\ U_{s,1}^{\dagger}\otimes U_{s',2}^{\dagger}
\end{eqnarray}
for all $s,s'$. From Theorem 1.1 of \cite{sarkar6}, we can express $B_0$ as 
\begin{eqnarray}
 B_0=\overline{B}_{0,ss'}\oplus E_{ss'}
\end{eqnarray}
where $E_{ss'}$ are unitary matrices.

Let us now denote Bob's local support of the states $\ket{\psi^i_s}$ as $V_{i,s}=\ $span$\{\proj{f_{0,i,s}},\proj{f_{1,i,s}}\}$ for all $i,s$. Further on, we will show that the supports $V_{i,l}, V_{i,l'}$ are orthogonal for any $l,l'$. For this purpose, we first express the product of the states $\ket{\psi^1_s}\ket{\psi^2_{s'}}$ as
\begin{eqnarray}
\ket{\psi^1_s}\ket{\psi^2_{s'}}=\frac{1}{2}\sum_{i,j=0,1}\ket{ij}\ket{f_{i,1,s}}\ket{f_{j,2,s'}}
\end{eqnarray}
which can equivalently be expressed using the Bell basis as
\begin{eqnarray}\label{STATE3}
\ket{\psi^1_s}\ket{\psi^2_{s'}}=\frac{1}{2}\sum_{i=1}^4\ket{\phi_i}\ket{g^i_{ss'}}
\end{eqnarray}
where $\ket{\phi_i}$ are given just above Eq. \eqref{A0} and 
\begin{eqnarray}\label{g1}
    \ket{g^1_{ss'}}=\frac{1}{\sqrt{2}}\left(\ket{f_{0,1,s}}\ket{f_{0,2,s'}}+\ket{f_{1,1,s}}\ket{f_{1,2,s'}}\right)\nonumber\\
    \ket{g^2_{ss'}}=\frac{1}{\sqrt{2}}\left(\ket{f_{0,1,s}}\ket{f_{0,2,s'}}-\ket{f_{1,1,s}}\ket{f_{1,2,s'}}\right)\nonumber\\
    \ket{g^3_{ss'}}=\frac{1}{\sqrt{2}}\left(\ket{f_{0,1,s}}\ket{f_{1,2,s'}}+\ket{f_{1,1,s}}\ket{f_{0,2,s'}}\right)\nonumber\\
    \ket{g^4_{ss'}}=\frac{1}{\sqrt{2}}\left(\ket{f_{0,1,s}}\ket{f_{1,2,s'}}-\ket{f_{1,1,s}}\ket{f_{0,2,s'}}\right)
\end{eqnarray}
Let us again utilize the relation \eqref{SOS2} and apply the state \eqref{STATE3} to it to observe that
\begin{eqnarray}
\sum_{i=1}^4\omega^i\ket{\phi_i}B_0^3\ket{g^i_{ss'}}=\sum_{i=1}^4\ket{\phi_i}\ket{g^i_{ss'}}.
\end{eqnarray}
Multiplying with $\bra{\phi_i}$ on both sides of the above expression gives us
\begin{eqnarray}\label{SOS4}
\omega^iB_0^3\ket{g^i_{ss'}}=\ket{g^i_{ss'}}\qquad \forall i.
\end{eqnarray}
As $B_0$ is unitary, we can conclude from the above formula \eqref{SOS4} that
\begin{eqnarray}\label{SOS5}
   \langle g^j_{ll'} \ket{g^i_{ss'}}=0 \qquad i\ne j
\end{eqnarray}
for any $i,j,l,l',s,s'$. Let us now consider Eq. \eqref{SOS5} with $l=s, j=1$ and expand it using \eqref{g1} to obtain the following conditions for $i=2,3,4$ as
\begin{subequations}
\begin{equation}\label{abcd1}
\langle{f_{0,2,l'}}\ket{f_{0,2,s'}}-\langle{f_{1,2,l'}}\ket{f_{1,2,s'}}=0
\end{equation}
\begin{equation}\label{abcd2}
\langle{f_{0,2,l'}}\ket{f_{1,2,s'}}+\langle{f_{1,2,l'}}\ket{f_{0,2,s'}}=0
\end{equation}
\begin{equation}\label{abcd3}
\langle{f_{0,2,l'}}\ket{f_{1,2,s'}}-\langle{f_{1,2,l'}}\ket{f_{0,2,s'}}=0.
\end{equation}
\end{subequations}
From Eqs. \eqref{abcd2} and \eqref{abcd3}, it is straightforward to observe that $\langle{f_{0,2,l'}}\ket{f_{1,2,s'}}=\langle{f_{1,2,l'}}\ket{f_{0,2,s'}}=0$. Let us now recall that $\ket{\psi^2_{l'}}$ and $\ket{\psi^2_{s'}}$ are orthogonal as they correspond to two different eigenvectors of $\rho_{AB}$ which gives us an additional condition
\begin{eqnarray}\label{abcd4}
\langle{f_{0,2,l'}}\ket{f_{0,2,s'}}+\langle{f_{1,2,l'}}\ket{f_{1,2,s'}}=0.
\end{eqnarray}
It is again straightforward to observe from \eqref{abcd1} and \eqref{abcd4} that $\langle{f_{0,2,l'}}\ket{f_{0,2,s'}}=\langle{f_{1,2,l'}}\ket{f_{1,2,s'}}=0$. Thus, the local supports $V_{2,s'}$ and $V_{2,l'}$ are orthogonal for any $s',l'$ such that $s'\ne l'$. Proceeding the same way as above, we can also conclude that the local supports $V_{1,s}$ and $V_{1,l}$ are orthogonal for any $s,l$ such that $s\ne l$. Consequently, the local supports $V_{ss'}=V_{1,s}\otimes V_{2,s'}$ are mutually orthogonal for any $s,s'$.

The local supports $V_{ss'}$ being mutually orthogonal imply that Bob's Hilbert space admits the following decomposition
\begin{equation}\label{block1}
    \mathcal{H}_B= \bigoplus_{s}\bigoplus_{s'}V_{ss'}=\bigoplus_{s}V_{1,s}\otimes\bigoplus_{s'}V_{2,s'}.
\end{equation}
As $\dim V_{1,s}=\dim V_{2,s'}=2$ for any $s,s'$, we can straightforwardly conclude that $\mathcal{H}_B=(\mathbbm{C}^2)_{B'_1}\otimes(\mathbbm{C}^2)_{B'_1}\otimes\mathcal{H}_{B''_{12}}$ where $\mathcal{H}_{B''_{1}}\otimes \mathcal{H}_{B''_{2}}$ for some Hilbert spaces
$\mathcal{H}_{B''_i}$.

The rest of the proof is exactly the same as step 3 in Theorem 1.2 of \cite{sarkar6}, which allows us conclude that there exist unitary transformations, $U_{i}:\mathcal{H}_B\rightarrow\mathcal{H}_B$,  such that
\begin{eqnarray}
(\mathbbm{1}_{A}\otimes U_1\otimes U_2)\rho_{AB}(\mathbbm{1}_{A}\otimes U_1^{\dagger}\otimes U_2^{\dagger})\qquad\qquad\qquad\nonumber\\=\proj{\phi^+}_{A_1B_1'}\otimes\proj{\phi^+}_{A_2B_2'}\otimes \rho_{B_1''B_2''},
\end{eqnarray}
where $\rho_{B_1''B_2''}$ denotes Bob's auxiliary state which is separable with 
\begin{eqnarray}
    U_i=\bigoplus_sU_{s,i} \qquad i=1,2
\end{eqnarray}
and
\begin{eqnarray}
U_1\otimes U_2 \,B_0\,U_1^{\dagger}\otimes U_2^{\dagger}=A_0\otimes \mathbbm{1}_{B_1''B_2''}.
\end{eqnarray}
This completes the proof.
\end{proof}
{\it{Randomness certification.}} An interesting application of the above self-testing statement is that the untrusted Bob's measurement device can be used to generate true randomness which is secure against any adversary. For this purpose, we consider an eavesdropper Eve who has access to Bob's laboratory. Consequently, we consider a state $\rho_{ABE}$ which is shared among Alice, Bob and Eve. As Eve's dimension is unrestricted, we can purify the state as $\ket{\psi_{ABE}}$ such that $\Tr_E\psi_{ABE}=\rho_{AB}$ where $\rho_{AB}$ is separable of the form as discussed below Eq. \eqref{SOS3}. 

Now, to certify whether the measurement outcomes as observed by Bob is truly random, we consider that Eve wants to guess the outcome of Bob's measurement. In order to do so,  she performs a measurement $Z=\{E_e\}$ on her share of the shared states. Here the outcome $e$ is Eve's best guess of Bob's outcome. However, any operation by Eve should not alter the statistics $\vec{p}=\{p(a,b)\}$ observed by Alice and Bob, that is,
\begin{eqnarray}
    p(a,b)=\bra{\psi}M_a\otimes N_b\otimes\I_E\ket{\psi}.
\end{eqnarray}
This is extremely important as the adversary Eve would like to remain invisible to Alice and Bob.

The number of random bits that can be securely generated from Bob's measurement is quantified as $H_{\min}=-\log_2 G(y,\vec{p})$ \cite{di4}, where $G(y,\vec{p})$ is known as the local guessing probability which can be computed as,
\begin{equation}\label{LGpr}
    G(\vec{p})=\sup_{S\in S_{\vec{p}}}\sum_{b}\bra{\psi} \I_A \otimes N_{b} \otimes E_b\ket{\psi},
\end{equation}
where $S_{\vec{p}}$ is the set of all Eve's strategies comprising of the shared states and her measurement that reproduce the probability distribution $\vec{p}$ as expected by Alice and Bob.

Let us now suppose that the swap-steering inequality \eqref{steein} is maximally violated
by $\vec{p}$. As proven above in Theorem \ref{Theo1M} this implies that the state shared by Alice, Bob and Eve up to local unitary operations is, $\ket{\psi_{{ABE}}}=\ket{\phi^+_{A_1B_1'}}\ket{\phi^+_{A_2B_2'}}\ket{\mathrm{\mathrm{aux}}_{\mathrm{B_{12}''E}}}$ as well as $N_b=\proj{\phi_b}\otimes\I_{\mathrm{B''_{12}}}$ where $\ket{\phi_b}$ are given above Eq. \eqref{A0}. Putting these states and measurement in the formula \eqref{LGpr} we obtain 
\begin{eqnarray}\label{guess1}
    G(\vec{p})=\sum_{b}\bra{\phi^+}\bra{\phi^+} (\I_A \otimes \proj{\phi_b})\ket{\phi^+}\ket{\phi^+} \nonumber\\ \bra{\mathrm{aux}}\I_{B_{12''}}\otimes E_b\ket{\mathrm{aux}}.
\end{eqnarray}
Now for all $b$, $\bra{\phi^+}\bra{\phi^+} (\I_A \otimes \proj{\phi_b})\ket{\phi^+}\ket{\phi^+}=1/4$ which allows us to conclude from \eqref{guess1} that
\begin{eqnarray}
     G(\vec{p})=\frac{1}{4}\sum_{b}\bra{\mathrm{aux}}\I_{B_{12''}}\otimes E_b\ket{\mathrm{aux}}=\frac{1}{4}.
\end{eqnarray}
Consequently, $-\log_2 G(\vec{p})=2$ bits of randomness can be certified from Bob's measurement outcomes using our self-testing scheme.

It is important here to note here that the generation of secure randomness is based on the assumption that the sources can only be correlated in a classical way. However, the adversary can always guess the outcomes of Bob if she manages to entangle the sources. For instance, (i) she can prepare both devices beforehand or (ii) she herself could perform an entangled measurement on the systems arriving on Bob's side and then send the outcome to Bob. This problem would persist in any security protocols involving two different constrained sources. However, the second type of attack (ii) can be avoided if Bob randomly chooses not to perform a measurement in some runs of the experiment. Since Eve is unaware of this fact, she would still entangle both sources and can be detected by Alice. 
It will be extremely interesting if Alice and Bob can perform some local operations on their subsystems to figure out whether the received subsystems are generated from separable sources or not.

\section{Conclusions}

In this work, we proposed a self-testing scheme in the 1SDI regime that can be used to certify the states generated by the sources to be maximally entangled and that the measurement performed by the untrusted party is the Bell basis measurement with both the parties performing a single measurement. Unlike most of the known self-testing schemes, we did not assume that any of the underlying states are pure or that the untrusted measurements are projective. This further allowed us to ensure that the randomness generated by the untrusted measurement device is completely random and the outcomes can not be guessed by any intruder.

Several follow-up questions arise from our work. The most interesting among them would be to find a fully DI scheme that can be used to certify randomness without inputs. Another direction would be to generalize the scheme presented in this work to certify an unbounded amount of randomness. This might be possible by either utilizing more than two sources or higher outcome measurements. In the self-testing scheme presented above, one can certify only the maximally entangled state. One might adapt the presented techniques to self-test more than one state in a single run of the experiment.

\begin{acknowledgments}
 This project was funded within the QuantERA II Programme (VERIqTAS project) that has received funding from the European Union’s Horizon 2020 research and innovation programme under Grant Agreement No 101017733.
\end{acknowledgments}

\input{ref_nlwist.bbl}
\end{document}

%% file: ref_nlwist.bbl
\providecommand{\noopsort}[1]{}\providecommand{\singleletter}[1]{#1}%